\documentclass[12pt]{amsart}
\usepackage{hyperref, amsmath,amssymb,amsfonts,amsbsy,mathtools,cite,setspace}
\usepackage{graphicx}
\usepackage{amssymb}
\usepackage{amsthm}
\usepackage{enumerate}
\usepackage{wasysym}
\usepackage{cite}
\usepackage{tikz}
\usepackage{mathtools}

\usepackage{braket}
\usepackage{color}

\def\ket#1{| #1 \rangle}
\def\bra#1{\langle #1 |}
\def\kb#1#2{|#1\rangle\!\langle #2 |}

\def\be{\begin{eqnarray}}
\def\ee{\end{eqnarray}}
\def\bee{\begin{eqnarray*}}
\def\eee{\end{eqnarray*}}

\newtheorem{defn}{Definition}

\newtheorem{thm}{Theorem}
\newtheorem{exa}{Example}

\newtheorem{rmk}{Remark}

\newcommand{\C}{{\mathbb C}}

\newcommand{\operp}{$\bigcirc$\kern-.91em{$\perp$}}

\newcommand{\spn}{\operatorname{span}}

\def\be{\begin{eqnarray}}
\def\ee{\end{eqnarray}}
\def\bee{\begin{eqnarray*}}
\def\eee{\end{eqnarray*}}
\def\ot{\otimes}

\begin{document}

\title[One-Way LOCC Indistinguishable Lattice States]{One-Way LOCC Indistinguishable Lattice States via Operator Structures}
\author[D.W.Kribs, C.Mintah, M.Nathanson, R.Pereira]{David W. Kribs$^{1,2}$, Comfort Mintah$^{1}$, Michael Nathanson$^3$, Rajesh Pereira$^{1}$}

\address{$^1$Department of Mathematics \& Statistics, University of Guelph, Guelph, ON, Canada N1G 2W1}
\address{$^2$Institute for Quantum Computing and Department of Physics \& Astronomy, University of Waterloo, Waterloo, ON, Canada N2L 3G1}
\address{$^3$Department of Mathematics and Computer Science, Saint Mary's College of California, Moraga, CA, USA 94556}

\begin{abstract}
Lattice states are a class of quantum states that naturally generalize the fundamental set of Bell states. We apply recent results from quantum error correction and from one-way local operations and classical communication (LOCC) theory, that are built on the structure theory of operator systems and operator algebras, to develop a technique for the construction of relatively small sets of lattice states not distinguishable by one-way LOCC schemes. We also present examples, show the construction extends to generalized Pauli states, and compare the construction to other recent work.
\end{abstract}



\date{March 25, 2020}

\maketitle

\section{Introduction}

A basic problem in quantum information theory is that of identifying a state from a set of known states on a composite quantum system, utilizing only quantum operations local to the individual subsystems \cite{bennett1999quantum,ghosh2004distinguishability,horodecki2003local,chefles2004condition}. Many problems in the subject can be seen as special cases of the so-called local operations and classical communication (LOCC) framework, such as quantum teleportation and data hiding \cite{Teleportation, terhal2001hiding,eggeling2002hiding}. The restricted problem of quantum state discrimination with only one-way LOCC operations, in which local operations are performed sequentially on the different subsystems, has been identified as a subproblem of central importance, with special emphasis placed on identifying small sets of indistinguishable states under the paradigm  \cite{Walgate-2000,Nathanson-2005,fan2004distinguishability,N13,cosentino2013small,yu2012four}.

An important class of quantum states, called {\it lattice states}, are a natural generalization of the fundamental Bell states and have been studied previously in the context of LOCC and Positive Partial Transpose (PPT) measurements; for instance in \cite{cosentino2013positive, cosentino2013small, yu2012four}. 
In this paper, we apply recently established results from one-way LOCC theory and quantum error correction \cite{kribs2016operator,kribs2019qeclocc}, that are based on a structural analysis of certain operator systems and operator algebras which arise in the LOCC framework, to the state discrimination problem for lattice states. Specifically, we develop a technique for the construction of relatively small sets of lattice states that are indistinguishable under one-way LOCC schemes. We also show how the technique can be applied to the related class of generalized Pauli states. Our approach gives added insight into exactly why such states are indistinguishable, in particular that this can be seen from properties of the underlying operator structures.

This paper is organized as follows. The next section contains preliminary details. Our constructions and a number of examples are included in the third section. We conclude in the fourth section by comparing our results to other work in the literature and we comment on the future outlook of the approach.

\section{Preliminaries}

We first describe lattice states, then review the required notions from operator systems and operator algebra theory. 

\strut 


The two-qubit Bell states are the canonical entangled basis of two qubits and are well-studied. Writing the two-qubit standard basis in the usual way ($\ket{ij} = \ket{i} \otimes \ket{j}$), the Bell states can be written as
\bee
\ket{\Phi_0} = \frac{ \ket{00} + \ket{11}}{\sqrt{2}} \qquad \ket{\Phi_1} = \frac{ \ket{01} + \ket{10}}{\sqrt{2}} \\ \ket{\Phi_2} = \frac{ \ket{01} - \ket{10}}{\sqrt{2}} \qquad \ket{\Phi_3} = \frac{ \ket{00} -\ket{11}}{\sqrt{2}} .
\eee
These states are naturally identified with the Pauli matrices by $\ket{\Phi_i} = (I \ot \sigma_i)\ket{\Phi_0}$ and where we write,
\bee
I = \sigma_0 =\left( \begin{matrix}  1 & 0 \cr 0 & 1 \end{matrix}\right) \qquad{ X} = \sigma_1 =\left( \begin{matrix}  0 & 1 \cr 1 & 0 \end{matrix}\right)  \\{Y} = \sigma_2 =\left( \begin{matrix}  0 & -i \cr i & 0 \end{matrix}\right) \qquad{ Z} = \sigma_3 =\left( \begin{matrix}  1 & 0 \cr 0 & -1 \end{matrix}\right) .
\eee
The Bell states generalize to the class of lattice states as follows. 

\begin{defn}
{\rm 
For $n \geq 1$, the class of {\it lattice states} $\mathcal L_n$ are given by $n$-tensors of the Bell states;
\bee
{\mathcal L}_n = \{ \ket{\Phi_i} : i \in \{0,1,2,3\} \}^{\ot n} \subseteq \C^{2^n } \otimes  \C^{2^n} .
\eee
}
\end{defn}

States in ${\mathcal L}_n$ are identified with elements of the set of tensor products of Pauli matrices, ${\mathcal P}_n = \{ \ot_{k = 1}^n \sigma_{i_k} \}$, via an extension of the Bell state identification above.

From a communications perspective, if we consider that, for each Bell state, one party called Alice controls each of the first qubits and another called Bob controls each of the second qubits, then the elements of ${\mathcal L}_n$ are maximally-entangled states between two $d$-dimensional ($n$-qubit) quantum systems, one controlled by each party, with $d = 2^n$. As noted in the introduction, lattice states can be seen as one natural generalization of the Bell states that have been studied previously in the context of LOCC and PPT measurements; for example, a set of four states in ${\mathcal L}_2$ that cannot be distinguished by any local measurements is given in \cite{yu2012four, cosentino2013small}.


\strut

Building on a characterization of one-way LOCC \cite{N13} and properties of operator structures \cite{davidson,pereira,pereira2}, a main result from the recent works \cite{kribs2016operator,kribs2019qeclocc} exhibits a connection between the ability to distinguish a set of quantum states with one-way LOCC on the one hand, and the necessary existence of a so-called `separating vector' for a related operator algebra on the other. This result will be improved upon for the current setting in the next section. Let us briefly review the preliminary notions required to do so.

An {\it operator system} $\mathfrak{S}$ is a subspace of operators on a given Hilbert space that is also self-adjoint (i.e., $A\in \mathfrak{S}$ if and only if $A^*\in \mathfrak{S}$) and contains the identity operator $I$. A (finite-dimensional) C$^*$-algebra $\mathfrak{A}$ is a self-adjoint subspace of operators  that is also closed under multiplication. Such algebras are always unitarily equivalent to a direct sum of full matrix algebras coming with multiplicities for each of the algebras \cite{davidson}; that is, there is a unitary transformation $U$ such that $U \mathfrak{A} U^* = \oplus_k (I_{m_k}\otimes M_{n_k})$ for some (unique) positive integers $m_k, n_k$ and where $M_{n}$ is the set of $n\times n$ complex matrices and $I_m$ is the $m \times m$ identity matrix.

A {\it separating vector} $\ket{\psi}$ for an algebra is characterized by the constraint: $A\ket{\psi} = B\ket{\psi}$ for two operators $A, B$ in the algebra if and only if $A = B$. Put another way,  $\ket{\psi}$ is a separating vector for an algebra $\mathfrak{A}$ if the mapping $A\mapsto A\ket{\psi}$ is injective on $\mathfrak{A}$.  This means that a necessary (but not sufficient) condition for an algebra $\mathfrak{A}$ of operators on a $d$-dimensional Hilbert space to have a separating vector is for $\dim (\mathfrak{A})\le d$, where $\dim (\mathfrak{A})$ is the dimension of $\mathfrak{A}$ as a linear space of operators ($\sum_k n_k^2$ with the form of $\mathfrak{A}$ as above).  A key result from the theory of operator algebras on the subject \cite{pereira2} tells us that $\mathfrak A$, in its unitarily equivalent direct sum form as above, has a separating vector if and only if $m_k \geq n_k$ for all $k$.

\section{Constructions of Relatively Small Sets of One-Way LOCC Indistinguishable States}
\subsection{Lattice States}
Recent work has applied the study of operator systems to the problem of identifying sets of bipartite states that cannot be distinguished using one-way LOCC. The following theorem (Theorem~3 in \cite{kribs2019qeclocc}) was proved as an extension of a result derived in \cite{kribs2016operator}.

\begin{thm}\label{thm1} Let $\{U_i \}$ be a set of operators on $\mathbb{C}^d$ and suppose the operator system $\mathfrak{S} = \spn \{ U_i^* U_j, I  \}$ is closed under multiplication and hence is a $\mathrm{C}^*$-algebra. If we let $\ket{\Phi}$ be a maximally entangled state on $\mathbb{C}^d \ot \C^d$, then the set of states $\mathcal S = \{ ( I\otimes U_i) \ket{\Phi} \}$ on $\mathbb{C}^d \otimes \mathbb{C}^d$ is distinguishable by one-way LOCC if and only if $\mathfrak{S}$ has a separating vector.
\end{thm}

Consider how this result might lend itself to potential application to sets of lattice states: the $U_i$ in that case could be taken as elements of ${\mathcal P}_n = \{ \ot_{k = 1}^n \sigma_{i_k} \}$, with its nice multiplicative properties. Indeed, we use this result and further structure of lattice states to construct comparatively small sets of states that cannot be perfectly distinguished with one-way LOCC.

\begin{exa}\label{3qubit example}
{\rm Consider the lattice states ${\mathcal L}_3  =  \{ \ket{\Phi_i} \otimes \ket{\Phi_j} \otimes \ket{\Phi_k}: i,j,k \in \{0,1,2,3\}\} \subseteq \C^{ 8} \otimes  \C^{ 8},$ which lie in bipartite 3-qubit Hilbert space. The elements in ${\mathcal L}_3$ correspond to tensor products of  Bell states $\{ \sigma_i \ot \sigma_j \ot \sigma_k \}$ which form a multiplicative group of order 64. These commute (modulo scalar multiples), and each element has order 2, so the group is isomorphic to $(\mathbb{Z}_2)^6$.

Consider the following set of matrices corresponding to a set of six specific lattice states:
\[
S = \{ I^{\ot 3}, Z \ot I \ot I, I \ot Z \ot I, I \ot I \ot Z , X \ot X \ot X, Y \ot Y \ot Y\}.
\]
Let $\mathfrak{A}$ be the algebra generated by the elements of $S$.  One can check that $\mathfrak{A}$ has dimension 16, which is bigger than $d = 8$, so $\mathfrak{A}$ has no separating vector. The operator system of interest $\mathfrak{S}$ is contained in $\mathfrak{A}$.

For each pair $\{i,j\}$, the product $U_i^*U_j \in \mathfrak{S} \subset \mathfrak{A}$. For $i \ne j$, these pairwise products are all distinct, implying that $\dim \mathfrak{S} = 1 + {|S| \choose 2} = 16$. This means that $\mathfrak{S} = \mathfrak{A}$, and $\mathfrak{S}$ is an algebra that has no separating vector.  Theorem \ref{thm1} tells us that  is not possible to distinguish the corresponding quantum states with one-way LOCC.
}\end{exa}

We generalize this construction to create families of small sets of states that cannot be distinguished with one-way LOCC, allowing us to state the following theorem.

\begin{thm}\label{LOCCLatticeExample}
For every $n > 1$ and $d = 2^n$, there exist sets of $m$ lattice states in $\C^d\ot\C^d$ that are not distinguishable with one-way LOCC, where
\be
m = \left\{ \begin{array}{ll} 2\sqrt{2d}-1 & \mbox{ if  $n$  is odd} \cr 3\sqrt{d}-1 & \mbox{ if  $n$  is even.} \end{array} \right.
\ee
\end{thm}

\begin{proof} The lattice states $\mathcal L_n$, where $d = 2^n$, are in one-to-one correspondence with $n$-tensor products of Pauli matrices ${\mathcal P}_n = \{ \ot_{k = 1}^n \sigma_{i_k} \}$.  This set has cardinality $4^n$, and under usual multiplication modulo the scalar matrices, it is isomorphic to the direct product group $(\mathbb{Z}_2)^{2n}$, where $\mathbb{Z}_2$ is the additive group $\{ 0,1\}$ modulo 2.

Consider any subset $S_0 \subseteq {\mathcal P}_n$ of size $(n+1)$ such that the algebra $\mathfrak{A}$ generated by the elements of $S_0$ 
has dimension $2^{n+1}$ as a linear space.
Split $S_0$ into two disjoint sets of $S_1$ and $S_2$ of size $k$ and $(n+1-k)$. Let $G_i$  be the multiplicative subgroup of ${\mathcal P}_n$ generated by $S_i$, for $i = 1,2$; and let $G$ be the group $G = G_1G_2 \cong (\mathbb{Z}_2)^{n+1}$. Then $\mathfrak{A}$ is also the algebra generated by the elements of $G$.

Considering our subgroups $G_1$ and $G_2$ as sets, let $S = G_1 \cup G_2$. Every element in $G$ can be written as a product of two elements in $S$, which implies that the smallest operator system containing $\{ U_1U_2: U_1, U_2 \in S\}$ is actually just  $\mathfrak{A}$, and hence we are in a situation in which Theorem~\ref{thm1} applies. Since the dimension of $\mathfrak{A}$ satisfies $\dim(\mathfrak{A}) = 2^{n+1} > d$, the algebra has no separating vector, and thus it follows from Theorem~\ref{thm1} that the lattice states $\{(I\otimes U)\ket{\Phi} : U\in S \}$ are not distinguishable by one-way LOCC.




Finally, we note that the size of $S$ is $|S| = |G_1| + |G_2| -1 = 2^k + 2^{n+1-k} -1$, since they contain no overlap except the identity. We can minimize the size of $S$ when $k = \frac{n+1}{2}$ or $k = \frac{n}{2}$, depending on the parity of $n$, and the result follows. 
\end{proof}

We give a concrete example of this construction.

\begin{exa}\label{General Construction Example}
{\rm For a general example in $\mathcal L_n$, we can set
\bee
S_1 &=&\{ I^{\ot i} \ot Z \ot I^{\ot n-i-1}\}_{i = 0}^{k-1}  \\ S_2 &=&\{ I^{\ot i} \ot Z \ot I^{\ot n-i-1}\}_{i = k}^{n-1}  \cup \{ X^{\ot n}\}.
\eee
It is easy to check that the algebra generated by $S_1$ has dimension $2^k$; the algebra generated by $S_2$ has dimension $2^{n+1-k}$; and the algebra generated by their union has dimension $2^{n+1}$. This gives us our set
\begin{eqnarray*}
S = \left( \{ I ,Z \}^{\ot k}\ot I^{\ot (n-k)} \right) &\cup & \left( I^{\ot k} \ot \{ I ,Z \}^{\ot n-k}  \right) \\ &\cup& \left( { X}^{\ot k} \ot \{ X,Y \}^{\ot n-k}  \right)
\end{eqnarray*}
with $|S| = 2^k + 2^{n-k+1} - 1$, which is minimized when $k = \lfloor \frac{n}{2}+1 \rfloor$ and $|S| \in  \{2\sqrt{2d}-1, 3\sqrt{d}-1\}$.
}
\end{exa}

\begin{rmk}
{\rm It is worth noting that Example~\ref{General Construction Example} is minimal, in the sense that its $m$ states are not distinguishable with one-way LOCC but that we can perfectly distinguish $(m-1)$ of them. If Alice and Bob measure each of their qubits in the eigenbasis of $Y = \sigma_2$, then they can perfectly distinguish the states in $S$ except for $Z^{\otimes k} \otimes I^{\otimes n+1-k}$ and $X^{\otimes k} \otimes Y^{\otimes n+1-k}$, which will give the same outcomes. Thus, removing either state from $S$ gives a set of $(m-1)$ states that can be perfectly distinguished with one-way LOCC.
}
\end{rmk}

\subsection{Generalized Pauli States}
We can also extend the construction to the class of generalized $d \times d$ Pauli matrices, which are given by
\bee X = \sum_{i = 0}^{d-1} \kb{i+1}{i} \quad \mathrm{and} \quad Z =\sum_{i = 0}^{d-1} \omega^i \kb{i}{i}, \eee
where $\omega$ is a primitive $d$th root of unity and $\ket{d}\equiv\ket{0}$ in $X$.

A comparable result from the literature on this class is the main result of \cite{wangetal2016}, which for $d \geq 4$ constructs an orthogonal set of generalized Pauli matrices with $\lceil 3\sqrt{d}\rceil -1$ maximally entangled states in $\mathbb{C}^d \times \mathbb{C}^d$ that is one-way LOCC indistinguishable. The size of this set of maximally-entangled states matches that of our set in Theorem \ref{LOCCLatticeExample} when $d$ is a power of 4, and it seems worth exploring possible connections between them. The construction in \cite{wangetal2016} is purely computational. Our construction below gives additional insight into exactly why such states are not distinguishable; namely, indistinguishability follows from certain identifiable features of underlying operator structures. In the process, we are able to identify a smaller subset of generalized Pauli matrices with the desired property.

\begin{thm}\label{LOCCGeneralizedPauliExample}
For every $d \ge 2$, there exist sets of $m$ generalized Pauli states in $\C^d\ot\C^d$ that are not distinguishable with one-way LOCC, where
\be
m = 4\left\lceil \sqrt{\frac{d}{2}} \right\rceil -1 . 
\ee
\end{thm}

\begin{proof}
As before, the proof is given by construction. For any fixed $k,l$, we can define the following subsets of the generalized Pauli matrices:
\bee
S_1 &=& \{ I,X,X^2, X^3, \ldots X^{k-1} \} \\ S_2 &=& \{ I,X^k,X^{2k}, X^{3k}, \ldots X^{(l-1)k} \} \\ S_3 &=& S_2Z =  \{ Z,X^kZ,X^{2k}Z, X^{3k}Z, \ldots X^{(l-1)k}Z \}.  \eee
We claim that the set of states $S = S_1 \cup S_2 \cup S_3$ cannot be distinguished with one-way LOCC if $kl \ge d$.

As above, we are interested in the operator system that contains the pairwise products in $S$; as well as any algebra contained in it. We can also define $\mathfrak{R}$ to be the linear span of the products $\{ U_i^*U_j : U_i \in S_1, U_j \in S_2\}$. A little thought shows that if $kl \ge d$, then $\mathfrak{R}$ is the set of linear combinations of powers of $X$, which is the $d$-dimensional algebra of matrices that are diagonal in the eigenbasis of $X$.  In this case,
\bee
\mathfrak{S} &:=& \spn \{ U_i^*U_j : U_i, U_j \in S\} \\&=&  \spn \{ RZ^j: R \in \mathfrak{R}, j \in \{-1,0,1\}\} .
\eee

As stated, $\mathfrak{R}$ is the set of linear combinations of powers of $X$. Let $\ket{\varphi_0}$ and  $\ket{\varphi_1}$ be the eigenvectors of $X$ with eigenvalues $1$ and $\omega$, respectively. Then
\bee
\kb{\varphi_0}{\varphi_0} = \frac{1}{d} \sum_{k = 0}^{d-1} X^k 
\quad \mathrm{and} \quad
 \kb{\varphi_1}{\varphi_1} = \frac{1}{d} \sum_{k = 0}^{d-1} \omega^{-k}X^k  
\eee
both belong to $\mathfrak{R}$.
Noting that $Z\ket{\varphi_0} = \ket{\varphi_1}$, we have $\kb{\varphi_0}{\varphi_1}  \in \mathfrak{R}Z$ and $\kb{\varphi_1}{\varphi_0}  \in \mathfrak{R}Z^*$. This means that 
\bee
{\mathfrak A} = \spn \{ \kb{\varphi_0}{\varphi_0}, \kb{\varphi_1}{\varphi_0}, \kb{\varphi_0}{\varphi_1} ,\kb{\varphi_1}{\varphi_1}  \} \subseteq \mathfrak{S},
\eee 
and ${\mathfrak A}$ is an algebra that is isomorphic to $M_2$, which has no separating vector. Since it is a subset of $\mathfrak{S}$, then using Theorem~1 in \cite{kribs2016operator} (the precursor result to Theorem~1 presented above) tells us that the set $S$ cannot be distinguished with one-way LOCC.


The size of our set $S$ is given by $m = k+2l-1$. The example in  \cite{wangetal2016} used $k = l = \lceil \sqrt{d}\rceil$, giving $m = 3\lceil \sqrt{d}\rceil-1$. We can make $m$ smaller by minimizing the quantity $k + 2l$ subject to $kl \ge d$. The absolute minimum value varies a little since $k$ and $l$ must be integers, but it will always suffice to set $l = \left\lceil \sqrt{\frac{d}{2}} \right\rceil$ and $k = 2l$. This gives the desired value $m = 4\left\lceil \sqrt{\frac{d}{2}} \right\rceil -1$.
\end{proof}

\begin{rmk}
{\rm
We note that when $d = 2^n$ for an odd value of $n$, then $m = 4\sqrt{\frac{2^n}{2}} -1 = 2\sqrt{2d}-1$, matching the size in Theorem \ref{LOCCLatticeExample}.
}
\end{rmk}

\begin{rmk}
{\rm 
A final observation is that in the case that $d$ is even, if we keep $S_1$ and $S_2$ in the example but set $S_3 = S_2Z^{d/2}$, then $\mathfrak{S}$ is a $2d$-dimensional algebra, which automatically has no separating vector and thus $S$ cannot be distinguished with one-way LOCC.}
\end{rmk}

\section{Outlook and Conclusions}
This work adds to the growing body of results and constructions in the subject of LOCC state distinguishability based on the analysis of operator structures, in particular expanding and building on results in \cite{kribs2016operator,kribs2019qeclocc} for the cases of lattice states and generalized Pauli states. We conclude by discussing how our approach compares to other related work in the literature. 

The size of the sets of maximally-entangled states in our examples grow with dimension of the Hilbert space. By contrast, we note that there exist sets of only {\it three} orthogonal maximally-entangled states which are not distinguishable with one-way LOCC in arbitrarily high dimensions \cite{N13, tian2016general}. However, in these constructions, the corresponding matrices were direct sums of generalized Bell states that were out of phase with each other by a constant complex multiple. They lack the algebraic structure of either the lattice states or the generalized Pauli states. It is also true in these cases that we can usually find vectors $\ket{\phi}$ with $\bra{\phi} U_i^*U_j \ket{\phi} = 0$ whenever $i \ne j$; however, we cannot complete it to a measurement. The examples in this paper  are sets for which no such $\ket{\phi}$ exists. 

If, instead of one-way LOCC, we are allowed to use PPT measurements, then our approach is much more powerful and, in fact, any set of $m$ orthogonal maximally-entangled states can be perfectly distinguished with PPT measurements when $m \le \frac{d}{2}+1$ \cite{N13}. This means that the examples we construct here occupy the space between the two paradigms, being distinguishable with PPT measurements but not with one-way LOCC.  It is not immediately apparent whether the states in these examples are distinguishable using full LOCC operations. 

The examples and constructions presented here suggest the possibility of further generalizations and applications that make use of the operator structure approach in the context of LOCC state distinguishability, and we plan to continue these investigations elsewhere.

\strut

{\noindent}{\it Acknowledgements.} D.W.K. was partly supported by NSERC and a University Research Chair at Guelph. C.M. was partly supported by a Mitacs Accelerate internship held with the African Institute for Mathematical Sciences (AIMS). R.P. was partly supported by NSERC. M.N. acknowledges the ongoing support of the Saint Mary's College Office of Faculty Development, and would like to thank Andrew Conner for helpful conversations.

\bibliographystyle{plain}

\bibliography{MNBibfile}

\end{document}